\documentclass[submission,copyright,creativecommons]{eptcs}
\providecommand{\event}{COS 2013} % Name of the event you are submitting to
\usepackage{breakurl}             % Not needed if you use pdflatex only.
\usepackage{amsmath} 
\usepackage{amsthm} 
\usepackage{amssymb} 
\usepackage{cleveref} 
\usepackage{tikz} 
\usepackage{wrapfig} 
\usepackage[final]{listings}
\crefname{lstlisting}{listing}{listings}
\crefname{axiom}{}{}
\newcommand\gref[1]{\nameref{#1} \nameCref{#1}}
\renewcommand\gref\Cref
\newcommand\Gref\Cref

\usepackage{bussproofs}
\usepackage{txfonts}
\usepackage{ifthen}
\usepackage{tikz}
\usetikzlibrary{patterns}
\usetikzlibrary{decorations.pathmorphing}

\providecommand\optionalSubscript[1]{\ifthenelse{\equal{#1}{}}{}{_{#1}}}
\providecommand\optionalSuperscript[1]{\ifthenelse{\equal{#1}{}}{}{^{#1}}}

\newcommand\fa{A}

\newcommand\ltrue\top
\newcommand\lfalse\bot

\newcommand\limply\to 

\newcommand\quant[2]{#1 #2.\ }

\newcommand\qforall{\quant\forall}
\newcommand\qexists{\quant\exists}

\tikzstyle{label} = [fill=white,inner sep=1pt]
\newcommand\IfNotEmpty[2]{\ifthenelse{\equal{#1}{}}{}{#2}}

\newcommand\PrDer[3]{
\begin{tikzpicture}[baseline=0.3cm,scale=1]
  \IfNotEmpty{#3}{ % branch name
    \draw decorate [decoration={snake,amplitude=2pt,segment length=14pt}] {(0,0.2) --  node[above,label] {\(#3\)} (0,1.25)};
  }
  \draw (-0.5,1) -- (-0.1,0) -- (0.1,0) -- (0.5,1);
  \IfNotEmpty{#1}{ % derivation name
    \node[above,label] () at (0,0.2) {\(#1\)}; 
  } 
\end{tikzpicture}
}

\newcommand\PrAss[3][]{\AxiomC{\([#2\optionalSuperscript{#1}]^{#3}\)}}
\newcommand\PrAx[2][]{\AxiomC{\(#2\optionalSuperscript{#1}\)}}
\newcommand\PrUn[2][]{\UnaryInfC{\(#2\optionalSuperscript{#1}\)}}
\newcommand\PrBin[2][]{\BinaryInfC{\(#2\optionalSuperscript{#1}\)}}

\newcommand\PrLbl[2][]{\LeftLabel{\(#2\)}\ifthenelse{\equal{#1}{}}{}{\RightLabel{\(#1\)}}}

\newcommand\PrInf[1][]{\ifthenelse{\equal{#1}{}}{%
\def\extraVskip{-2pt}\noLine\UnaryInfC\vdots\noLine\def\extraVskip{2pt}}{%
\noLine\UnaryInfC{\(#1\)}}}
\newcommand\RuleName[3][]{#2\mathrm{#3}\IfNotEmpty{#1}{_\mathrm{#1}}}

\newcommand\PrImplyI[2][]{%
  \PrLbl[#1]{\RuleName\limply{I}}%
  \PrUn{#2}%
}
\newcommand\PrImplyE[2][]{%
  \PrLbl[#1]{\RuleName\limply{E}}%
  \PrBin{#2}%
}

\newcommand\PrExistsE[2][]{%
  \PrLbl[#1]{\RuleName\exists{E}}%
  \PrBin{#2}%
}
\newcommand\PrForallI[2][]{%
  \PrLbl[#1]{\RuleName\forall{I}}%
  \PrUn{#2}%
}
\newcommand\PrForallE[2][]{%
  \PrLbl[#1]{\RuleName\forall{E}}%
  \PrUn{#2}%
}

\theoremstyle{plain}
\newtheorem{theorem}{Theorem} 
\newtheorem{lemma}{Lemma}

\newtheorem*{theorem*}{Theorem}

\title{Interpreting a Classical Geometric Proof with Interactive Realizability}
\author{Giovanni Birolo
\institute{Department of Mathematics\\University of Turin\\Italy}
\email{giovanni.birolo@gmail.com}
}
\def\titlerunning{Interpreting a Classical Geometric Proof with Interactive Realizability}
\def\authorrunning{Giovanni Birolo}
\begin{document}
\maketitle
\newcommand\HA{\mathsf{HA}} % Heyting Arithmetic

\newcommand\riA{\alpha}
\newenvironment{legend}{\small \it}{}

\newcommand\N{\mathbb N}
\newcommand\Q{\mathbb Q}
\newcommand\R{\mathbb R}
\newcommand\EM{{\mathsf{EM}_1}}

\newcommand\lfa{f}
\newcommand\lva{x}
\newcommand\lvb{y}
\newcommand\na{n}
\newcommand\nb{m}
\newcommand\ia{i}
\newcommand\ib{j}
\newcommand\ic{k}
\newcommand\id{l}

\newcommand\pa{P}
\newcommand\pb{Q}
\newcommand\pc{R}
\newcommand\pd{S}

% pointAt: name, position, label, label positioning 
\newcommand\pointAt[4]{%
  \coordinate (#1) at (#2);
  \draw[fill] (#1) circle (0.05);
  \node[#4] at (#1) {\(#3\)};
}
% lineBetween: point name, point name
\newcommand\lineBetween[2]{%
  \draw [shorten >= -1cm, shorten <=-1cm] (#1)--(#2);
}

\begin{abstract}
We show how to extract a monotonic learning algorithm from a classical proof of a geometric statement by interpreting the proof by means of interactive realizability, a realizability sematics for classical logic. 

The statement is about the existence of a convex angle including a finite collections of points in the real plane and it is related to the existence of a convex hull. 
We define real numbers as Cauchy sequences of rational numbers, 
therefore equality and ordering are not decidable. 
While the proof looks superficially constructive, 
it employs classical reasoning to handle undecidable comparisons between real numbers, 
making the underlying algorithm non-effective. 

The interactive realizability interpretation transforms the non-effective linear algorithm described by the proof into an effective one that uses backtracking to learn from its mistakes. 
The effective algorithm exhibits a ``smart'' behavior, performing comparisons only up to the precision required to prove the final statement. 
This behavior is not explicitly planned but arises from the interactive interpretation of comparisons between Cauchy sequences. 
\end{abstract}

\section{Introduction}

Interactive realizability is a realizability semantics that extends the Brouwer-Heyting-Kolmogorov interpretation to (sub-)classical logic, more precisely to first-order intuitionistic arithmetic (Heyting Arithmetic, \(\HA\)) extended by the law of the excluded middle restricted to \(\Sigma^0_1\) formulas (\(\EM\)), a system motivated by its applications in proof mining. 
It was introduced by Berardi and de'Liguoro in \cite{berardidL08}. 

We use interactive realizability in order to study the computational content of a classical proof of the following geometric statement. 
\begin{theorem*}[Convex Angle] \ 

  \begin{minipage}{0.65\linewidth}
    We have a finite set of at least three points in the real plane \(\R^2\) such that no three points are on the same line. 
    Then there exist distinct points \(\pa, \pb\) and \(\pc\) such that: 
    \begin{itemize}
      \item all other points \(\pd\) are inside \(\widehat{\pb\pa\pc}\), 
      \item the angle \(\widehat{\pb\pa\pc}\) is convex, that is, less than \(\pi\). 
    \end{itemize} 
  \end{minipage}\begin{minipage}{0.30\linewidth}
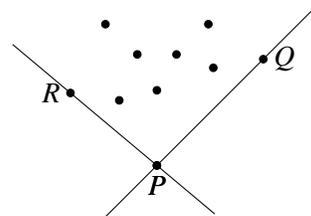
    \begin{center}
      \begin{tikzpicture}[>=latex] 
        \pointAt{A}{0,0}\pa{below}
        \pointAt{B}{45:2}\pb{right}
        \pointAt{C}{140:1.5}\pc{left}
        \lineBetween{A}{B}
        \lineBetween{A}{C}
        \pointAt{A}{0,0}\pa{below}
        \pointAt{D}{60:1.5}{}{}
        \pointAt{D}{70:2}{}{}
        \pointAt{D}{80:1.5}{}{}
        \pointAt{D}{90:1}{}{}
        \pointAt{D}{100:1.5}{}{}
        \pointAt{D}{110:2}{}{}
        \pointAt{D}{120:1}{}{}
      \end{tikzpicture}
    \end{center}
  \end{minipage}
\end{theorem*} 

We choose this particular statement because we have a proof of it that looks algorithmic and can be easily visualized. 
\Gref{thm:bounding_angle} can be thought of as weakened version of the existence of the convex hull of a finite set of points. 

As we said the proof we choose as example looks constructive, using only decidability of ordering over real numbers. 
However, it is well known that there is no effective ordering on the real numbers. 
In our encoding of the real numbers, totality of the ordering on the recursive reals is equivalent to \(\EM\). 
Since the proof needs the ordering to be total, it needs \(\EM\). 
Due to the low logical complexity of excluded middle which is used, the proof may be interpreted with a simple case of interactive realizability.

We show how interactive realizability can be applied and what it can tell us about the computational content of the proof. 
What we get is an algorithm that, instead of comparing real numbers, makes an arbitrary guess about which one is smaller.
If later it becomes apparent that the guess is wrong the algorithm retracts the choice it made since it can now make an informed decision about that particular comparison. 
Then the algorithm performs comparisons only when needed and only up to the required precision. 

Thus we see how a simple classical proof which performs comparisons between real numbers is interpreted as a learning algorithm which uses ``educated guesses'' 
in order to avoid non effective operations. 
This non-trivial behavior is not explicit in the classical proof, but follows from the definition of ordering on Cauchy sequences by means of the interactive realizability interpretation. 

In the present work, our main goal is to showcase interactive realizability and the backtracking algorithms it produces through a non-trivial example. 
For this reason, we chose to present interactive realizability as a proof interpretation technique rather than as a realizability semantics, in order to concentrate on the example and its computational interpretation without being bogged down in technical details.  
A more comprehensive treatment of interactive realizability can be found for instance in \cite{aschieriB10}.

\section{Real Numbers}

In this section we present our treatment of real numbers in Heyting Arithmetic. For a more in depth treatment of real numbers from a constructive view point see \cite{troelstraVD88}.

% rational numbers
\newcommand\eqQ{=_\Q}
\newcommand\ltQ{<_\Q}
\newcommand\gtQ{>_\Q}
\newcommand\leQ{\le_\Q}
\newcommand\geQ{\ge_\Q}
\newcommand\Qplus{+_\Q}
\newcommand\Qminus{-_\Q}
\newcommand\Qprod{\cdot_\Q}
\newcommand\Qzero{0_\Q}
\newcommand\lqa{q}
\newcommand\lqb{p}
There are many ways of encoding integer and rational numbers in \(\HA\) and defining primitive recursive operations and predicates on them. 
In the following we assume that we have any such encoding and that we have decidable equality \(\eqQ\) and ordering \(\ltQ, \leQ\) and effective operations \(\Qplus, \Qprod\).
We use the variables \(\lqa\) and \(\lqb\) for rationals. 

%\usetikzlibrary{matrix}
%\begin{tikzpicture}
%  \matrix(q) [matrix of math nodes] {
%    \frac{0}{1} & \frac{1}{1} & \frac{2}{1} & \frac{3}{1} \\
%    \frac{0}{2} & \frac{1}{2} & \frac{2}{2} & \frac{3}{2} \\
%    \frac{0}{3} & \frac{1}{3} & \frac{2}{3} & \frac{3}{3} \\
%    \frac{0}{4} & \frac{1}{4} & \frac{2}{4} & \frac{3}{4} \\
%  };
%  \path
%  (q-1-1.south west) 
%  edge (q-1-1.north east)
%  edge [dotted] (q-1-2.south west)
%  edge (q-2-1.north east)
%  edge [dotted] (q-1-3.south west)
%  edge (q-3-1.north east);
%\end{tikzpicture}

%\subsection{Cauchy Sequences}
There are many equivalent ways of defining the real numbers from the rational numbers. 
The best known are the definition of the reals as equivalence classes of Cauchy sequences and as Dedekind cuts.
We follow the first approach. 

\newcommand\lrea{r}
\newcommand\lreb{s}
\newcommand\lrec{t}
\newcommand\lpa{k}
\newcommand\lpb{l}
A sequence of rationals \( \lrea : \N \to \Q \) is a \emph{Cauchy sequence} if the following holds:
\begin{equation} \label{eq:classical_cauchy}
  \qforall\lpa \qexists{\lpa_0} \qforall{\lpa_1, \lpa_2}
  | \lrea(\lpa_0 + \lpa_2) - \lrea(\lpa_0 + \lpa_1) | < \frac{1}{2^\lpa}. 
\end{equation}
While this sequence approximates a real number, it can do so very slowly. 
By means of classical reasoning, 
we can show that, from any Cauchy sequence, we can extract a fast-converging monotone sub-sequence. 
For this reason, instead of general Cauchy sequences, we can consider sequences of nested intervals with rational extremes whose length decreases exponentially. 
An interval is determined by its extremes, so we represent a sequence of intervals as a couple of sequences of rationals \(\lrea^-, \lrea^+\), representing the lower and higher extremes of the intervals respectively.
Then we require that \(\lrea^-\) is increasing and \(\lrea^+\) is decreasing (since the intervals are nested), that \(\lrea^-(\lpa)\) is lesser than or equal to \(\lrea^+(\lpa)\) (since they are the lower and higher extremes of a same interval) and their difference is smaller than \(2^{-\lpa}\). 
More precisely we say that \(\lrea^-\) and \(\lrea^+\) represent a real number when they satisfy the following condition, written as a \(\Pi^0_1\) formula:
\begin{equation} \label{eq:nested_cauchy}
  \begin{split}
    \qforall\lpa &(\lrea^-(\lpa) \leQ \lrea^+(\lpa)) \land (\lrea^-(\lpa) \leQ \lrea^-(\lpa+1)) \land \\ 
                 &\land (\lrea^+(\lpa) \geQ \lrea^+(\lpa+1)) \land (\lrea^+(\lpa) \Qminus \lrea^-(\lpa) \leQ 2^{(-\lpa)}).
  \end{split}
\end{equation} 
While the choice of the specific definition of real number is somewhat arbitrary, it is significant because it affects the logical properties (in particular the degree of undecidability) of the ordering on the reals. 

%\subsection{Order Predicate}
\newcommand\op{\text{OP}}
\newcommand\RuleNameOP[1]{\op\text{-{#1}}}

Now we can define an ``order predicate'' \(\op(\lrea, \lreb, \lpa)\), which can be thought of as a family of strict partial orders on the real numbers indexed by natural number \(\lpa\). 
More precisely, it is a formula that determines when the sequence of nested intervals \(\lrea\) is strictly lesser than \(\lreb\), at precision \(\lpa\). 
This happens when, at \(\lpa\), the higher extreme of an interval is strictly greater than the lower extreme of the other. Then, from that point forward, the intervals will be forever disjoint, since they are nested sequences. 
This allows us to write the order predicate as the formula:
\begin{equation} \label{eq:nested_op}
  \op(\lrea, \lreb, \lpa) \equiv \lrea^+(\lpa) \ltQ \lreb^-(\lpa), 
\end{equation} 
which is decidable in \(\lrea\) and \(\lreb\). 
Note that the definition of \(\op\) depends on that of real number. If we had used the classical definition of Cauchy sequence the order predicate would be the following \(\Pi_1^0\) formula:
\begin{equation} \label{eq:cauchy_op}
  \op'(\lrea, \lreb, \lpa) \equiv \qforall{\lpb} \lpb \ge \lpa \limply \lrea(\lpb) \ltQ \lrea(\lpb). 
\end{equation} 
This is very significant for our purposes: 
the order predicate in \eqref{eq:nested_op} is decidable in \(\lrea\) and \(\lreb\) (since the order on the rationals is), while in \eqref{eq:cauchy_op} it is only \emph{negatively decidable}. 
This means that we have an effective method to decide \eqref{eq:cauchy_op} when it is false, but not when it is true. 
%More precisely
% \eqref{eq:cauchy_op} is false we can look for a counterexample, that is, a natural number \(\lpb \ge \lpa\) such that \( \lrea(\lpb) \ltQ \lrea(\lpb) \).

We need \(\op\) to satisfy the following properties, written as rules:
\begin{equation} \label{eq:op_properties}
  \begin{array}{cc}
    \PrAx{\op(\lrea, \lreb, \lpa)}
    \PrLbl{\RuleNameOP{mon}}
    \PrUn{\op(\lrea, \lreb, \lpa+1)}
    \DisplayProof
    &
    \PrAx{\op(\lrea, \lrea, \lpa)}
    \PrLbl{\RuleNameOP{irrefl}}
    \PrUn\lfalse
    \DisplayProof
    \\[6mm]
    \PrAx{\op(\lrea, \lreb, \lpa)}
    \PrAx{\op(\lreb, \lrea, \lpb)}
    \PrLbl{\RuleNameOP{asym}}
    \PrBin\lfalse
    \DisplayProof
    &
    \PrAx{\op(\lrea, \lreb, \lpa)}
    \PrAx{\op(\lreb, \lrec, \lpb)}
    \PrLbl{\RuleNameOP{trans}}
    \PrBin{\op(\lrea, \lrec, \max(\lpa,\lpb))}
    \DisplayProof
  \end{array}
\end{equation}
The \(\RuleNameOP{mon}\) rule expresses a monotonicity property: when an comparison at a given precision can distinguish two approximations, then comparisons at greater precision should too. 
The other rules correspond to the standard axioms for a strict partial order: irreflexivity, asymmetry and transitivity. 

We verify that our definition of \(\op\) satisfies these properties. 
\begin{lemma} \label{thm:op_properties}
  The order predicate \(\op\) defined by \eqref{eq:nested_op} satisfies the properties given in \eqref{eq:op_properties}.
\end{lemma}
\begin{proof} 
  Omitted. The properties follow directly from the definition of \(\op\) as \eqref{eq:nested_op} and from our representation of real number as sequences of nested intervals \eqref{eq:nested_cauchy}. 
\end{proof} 

%\subsection{Order and Equality on the Real Numbers}

\newcommand\eqR{=_\R}
\newcommand\neqR{\neq_\R}
\newcommand\ltR{<_\R}
\newcommand\gtR{>_\R}
\newcommand\leR{\le_\R}
\newcommand\geR{\ge_\R}

We can now define order and equality on the reals.
It is noteworthy that, while we define order and equality in terms of \(\op\),  we never use the definition of \(\op\) itself in proving their properties. 
We only need the properties of \(\op\) we proved in \Cref{thm:op_properties}, 
thus we could proceed in the same way even if we had defined \(\op\) differently, as long as \Cref{thm:op_properties} holds. 

They are defined as follows:
\begin{align*}
  \lrea \ltR \lreb &\equiv \qexists\lpa \op(\lrea, \lreb, \lpa), &
  \lrea \leR \lreb &\equiv \qforall\lpa \lnot \op(\lreb, \lrea, \lpa), \\ 
  \lrea \neqR \lreb &\equiv \qexists\lpa \op(\lrea, \lreb, \lpa) \lor \op(\lreb, \lrea, \lpa), &
  \lrea \eqR \lreb &\equiv \qforall\lpa \lnot \op(\lrea, \lreb, \lpa) \land \lnot \op(\lreb, \lrea, \lpa). 
\end{align*}
Note that \(\ltR\) and \(\neqR\) are \(\Sigma^0_1\) formulas and \(\leR\) and \(\eqR\) are \(\Pi^0_1\) formulas. 

In order to prove \gref{thm:least_element}, which is needed in the proof of \gref{thm:bounding_angle}, 
we need to show some of the properties of the order \(\leR\). 
\begin{lemma}[Reflexivity, Semi-Transitivity and Totality of \(\leR\)]
  The following properties hold:
  \begin{align}
    \tag{reflexivity}\label[property]{prop:Rrefl} 
    \lrea \leR \lrea \\ 
    \tag{semi-transitivity}\label[property]{prop:Rtrans} 
    \lrea \ltR \lreb \land \lreb \leR \lrec \limply \lrea \leR \lrec, \\
    \tag{totality}\label[property]{prop:Rtot}
    \lrea \leR \lreb \lor \lreb \ltR \lrea. 
  \end{align}
\end{lemma}
\begin{proof}
  The first two properties follows from the corresponding properties of \(\op\). 
  The last is a classical tautology. 
  \begin{itemize}
    \item
      We omit the proof of reflexivity for reasons of space and irrelevance. 
    \item
      In order to prove this transitive property for mixed \(\ltR\) and \(\leR\) we have to show that \( \lrea \leR \lrec \equiv \qforall\lpa \lnot \op(\lrec, \lrea, \lpa),\) 
        assuming 
        \( \lrea \ltR \lreb \equiv \qexists\lpa \op(\lrea, \lreb, \lpa) \) and 
        \(\lreb \leR \lrec \equiv \qforall\lpa \lnot \op(\lrec, \lreb, \lpa)\).
      This follows by means of the \RuleNameOP{trans} rule: 
      \[
        \PrAx{\qexists\lpa \op(\lrea, \lreb, \lpa)}
        \PrAx{\qforall\lpa \lnot \op(\lrec, \lreb, \lpa)}
        \PrForallE{\lnot \op(\lrec, \lreb, \max(\lpa,\lpb))}
        \PrAss{\op(\lrec, \lrea, \lpa)}{1}
        \PrAss{\op(\lrea, \lreb, \lpb)}{2}
        \PrLbl{\RuleNameOP{trans}}
        \PrBin{\op(\lrec, \lreb, \max(\lpa,\lpb))}
        \PrLbl{\RuleNameOP{asym}}
        \PrImplyE\lfalse
        \PrExistsE[2]\lfalse
        \PrImplyI[1]{\lnot \op(\lrec, \lrea, \lpa)}
        \PrForallI{\qforall\lpa \lnot \op(\lrec, \lrea, \lpa)}
        \DisplayProof
      \]
    \item
      When \(\lrea\) and \(\lreb\) denote recursive real numbers, 
      totality is an instance of \(\EM\): 
      \[ 
        \lrea \leR \lreb \lor \lreb \ltR \lrea \equiv
        \qforall\lpa \lnot \op(\lreb, \lrea, \lpa) \lor \qexists\lpa \op(\lrea, \lreb, \lpa). \qedhere
      \]
  \end{itemize} 
\end{proof}
The proof is constructive apart from the last point, where we show that totality is actually an instance of \(\EM\). 
Note that only the reflexivity property is stated in the standard way, while transitivity and totality are written in non-standard forms.  
We chose these forms for two reasons: they are easier to prove and they are the exact forms we need in the proof of \gref{thm:least_element}. 

%\subsection{Variables for Real Numbers}

Until now we have used \(\lrea, \lreb\) and \(\lrec\) as metavariables for real numbers in an informal way. 
However, since we are working in the first-order language of arithmetic, our variables range only on natural numbers and not on functions. 
For our example we only need to address a finite but arbitrary number of real numbers, that is, we only need a countable quantity of them. 
Thus we can assume that we have a countable set of pairs of function symbols indexed by the natural numbers, say \( (\lfa^+_\na, \lfa^-_\na)_{\na \in \N}\). 
We assume that each pair satisfies the convergence condition \eqref{eq:nested_cauchy} and thus represents a real number. 
Then, \(\op\) can be formally defined as \( \lfa^+_\ia(\lpa) \ltQ \lfa^-_\ib(\lpa) \) where \(\ia\) and \(\ib\) are arithmetic terms. 
Thus each real numbers is represented by a natural number, namely its index. 
For convenience and consistency with the standard notation for real numbers, instead of writing \( \ia \ltR \ib \), we use the sugared version \(\lrea_\ia \leR \lrea_\ib\).

%\subsection{The Least Element Lemma}

Now we can reason about finite sets of real numbers as sets of indexes. 
In the next lemma, we shall work with the sets of real numbers indexed by initial segments of the natural numbers. 
We show the existence of a least element in each of these sets. 
The least element is actually a minimum, that is, the unique least element of the set. 
However, in order to prove \gref{thm:bounding_angle}
we do not need to show its uniqueness, just its existence. 
\begin{lemma}[Least Element] \label{thm:least_element}
  For any \(\na\), the real numbers \(\lrea_0, \dotsc, \lrea_\na\) have a least element with respect to \(\leR\). 
  More precisely: 
  % \footnote{
  %   We use the standard compact notation for bounded quantifications: 
  %   \begin{gather*}
  %     \qforall{\ib \le \na} \fa \text{ stands for } \qforall\ib \ib \le \na \limply \fa, \\ 
  %     \qexists{\ib \le \na} \fa \text{ stands for } \qexists\ib \ib \le \na \land \fa. 
  %   \end{gather*}
  % }
  \[
    \qforall\na \qexists{\ia \le \na} \qforall{\ib \le \na} \lrea_\ia \leR \lrea_\ib. 
  \]
\end{lemma}
\begin{proof}
  We proceed by induction on \(\na\). 
  \begin{description}
    \item[Zero case] 
      In the base case \(\na = 0\) and we have to prove that \( \qexists{\ia \le 0} \qforall{\ib \le 0} \lrea_\ia \leR \lrea_\ib \).
      Both \(\ia\) and \(\ib\) can only be \(0\); 
      thus we just have to check the condition \( \lrea_0 \leR \lrea_0 \), 
      which holds by reflexivity of \(\leR\). 
    \item[Successor case]
      In the inductive case we have to prove that 
      \(
        \qexists{\ia \le \na+1} \qforall{\ib \le \na+1} \lrea_\ia \leR \lrea_\ib
      \).
      By the inductive hypothesis, let \( \bar\ia \le \na \) be the index of the least element in \(\lrea_0, \dotsc, \lrea_\na\). 
      By totality of \(\leR\) we have two cases.
      \begin{description}
        \item[\( \lrea_{\bar\ia} \leR \lrea_{\na+1} \)]
          Then \(\bar\ia\) is the index of a least element in \( \lrea_0, \dotsc, \lrea_{\na+1} \),
          since \( \lrea_{\bar\ia} \leR \lrea_\ib \) when \(\ib = \na+1\) (since we are considering this case) and when \(\ib \le \na\) by inductive hypothesis. 
        \item[\( \lrea_{\na+1} \ltR \lrea_{\bar\ia} \)]
          Then \(\na+1\) is the index of a least element in \(\lrea_0, \dotsc, \lrea_{\na+1} \), 
          since \( \lrea_{\na+1} \leR \lrea_\ib\) when \(\ib = \na+1\) by reflexivity of \(\leR\) and when \(\ib \le \na\) by transitivity of \(\ltR\) and \(\leR\), since 
          \(
            \lrea_{\na+1} \ltR \lrea_{\bar\ia} \leR \lrea_\ib
          \)
          by inductive hypothesis. \qedhere
      \end{description}
  \end{description}
\end{proof}

\newcommand\prog[1]{``\lstinline"#1"''}
The proof looks constructive: its computational interpretation is the usual algorithm that finds the least element in a vector, by a simple recursion or by looping on its elements.
%We begin by taking the first element of the vector as a tentative least element and we compare it with the next element, replacing it 
We can write it as a recursive function \prog{rmin} in Haskell: 
\begin{lstlisting}[caption=The Least Element Program,label=prog:min]
rmin 0   = 0
rmin n   = if rle (rmin (n-1)) n
  then rmin (n-1)
  else n
\end{lstlisting}
where \prog{rle} is a boolean function that stands for \(\leR\), that is, it compares the reals indexed by its arguments. 
The problem is that we are unable to write \prog{rle} as a terminating program.
The closest approximation would be the following unfounded recursion:
\begin{lstlisting}[caption=The Lesser or Equal Program,label=prog:rle]
rle i j        = rle_urec 0 i j
rle_urec k i j = if op j i k end 
  then False
  else rle_urec (k+1) i j
\end{lstlisting}
where \prog{op} is a total boolean function that stands for the order predicate \(\op\). 
We can assume that  \prog{op} terminates for any input since \(\op\) is decidable. 
The problem is that \(\leR\) is total only classically. 
More precisely, totality is an instance of \(\EM\) because \(\leR\) is a \(\Pi^0_1\) formula and thus negatively decidable.
This can bee seen concretely in the program for \prog{rle}: 
\prog{rle i j} only halts (returning \prog{False}) if \prog{op j i k} is true for some \(k\), that is, if and only if \(\lrea_\ia \leq \lrea_\ib\) is false. 
On the other hand, when \(\lrea_\ia \leq \lrea_\ib\) is true there is no such \(k\) and the evaluation of \prog{rle i j} will never halt: 
\prog{True} does not even appear in the program.  
This is the general behavior of an algorithm that computes a negatively decidable predicate: when the predicate is false it halts with the correct answer and when the predicate is true it does not halt. 

For positively decidable predicates we have the dual behavior. 
For instance, in the case of \(\ltR\) which is defined by a \(\Sigma^0_1\) formula and thus positively decidable, the decision procedure can be written as: 
\begin{lstlisting}[caption=The Lesser Than Program,label=prog:rlt] 
rlt i j        = rlt_urec 0 i j
rlt_urec k i j = if op i j k  
  then True
  else rlt_urec (k+1) i j
\end{lstlisting}
The program is very similar to the previous one, the only noteworthy changes are the order of the argument given to \prog{op} and the fact that the only possible return value is \prog{True} instead of \prog{False}. 
It only halts (returning \prog{True}) if \prog{op i j k} is true for some \(k\), that is, when \(\lrea_\ia \leq \lrea_\ib\) is true.

\section{The Interactive Interpretation of the Least Element Lemma} 

We have seen why the naive way of extracting a program from proofs fails in the case of \gref{thm:least_element}. 
Now we give the interactive interpretation of \gref{thm:least_element}. 
  Since we are working in \(\HA+\EM\), any proof can be thought of as a constructive proof with open assumptions, that are the instances of \(\EM\) that are used in the proof.
  The interactive realizability interpretation follows the standard BHK interpretation for the constructive parts, so we will concentrate on the interpretation of the \(\EM\) instances. 

%The interactive interpretation is guided by the classical parts of the proof, namely, by the \(\EM\) instances. 
The only instances of \(\EM\) in the proof are those used to deduce the totality property:
\begin{equation} \label{eq:em_tot_instance} 
  \lrea_\ia \leR \lrea_\ib \lor \lrea_\ib \ltR \lrea_\ia. 
\end{equation}
The left disjunct, which we call the \emph{universal disjunct}, is \(\Pi^0_1\) and negatively decidable, while the right one, the \emph{existential disjunct} is \(\Sigma^0_1\) and positively decidable. 
Moreover universal disjunct and negation of the existential disjunct are classically equivalent. 
We say that a formula is \emph{concrete} when it is closed and atomic. 

A naive attempt to give a computational content to an \(\EM\) instance fails, because in general \(\EM\) instances are undecidable. 
Interactive realizability proposes a way to side-step this problem. 
This is possible since it is not true that the computational interpretation of a proof using instances of \(\EM\) necessarily needs to decide these instances. 
Consider the case of totality of the order on the real numbers. 
The universal disjunct is \(\lrea_\ia \leR \lrea_\ib \equiv \qforall\lpa \lnot \op(\lrea_\ib, \lrea_\ia, \lpa)\). 
Being an universally quantified statement, it proves infinite instances \(\lnot \op(\lrea_\ib, \lrea_\ia, \lpa)\), one for each natural number \(\lpa\). 
A proof that uses totality may need all this infinite information or (for example, when proving a simply existential statement) may only need a finite quantity of these instances. 
In the second case, we can avoid the problem of effectively deciding the \(\EM\) instance. 
We only need to decide those instances that are actually used in the proof. 
This is possible, since each instance is decidable (being a quantifier free formula) and we assumed there is a finite quantity of them. 
Interactive realizability takes advantage of this fact and gives a procedure to determine which instances of the universal disjunct are needed and to iteratively decide them. 

% comparison with BHK interpretation
The interactive interpretation is a ``relaxation'' of the BHK interpretation. 
%it decides disjunctions and produce witness for existential quantifications. 
In the BHK interpretation the decision of a disjunction effectively selects a true disjunct, in the interactive case instead of a decision we have a sort of ``educated guess''. 
Therefore, while \(\EM\) cannot be realized by the BHK interpretation since there is no effective procedure to decide it, the interactive interpretation can because it yields a weaker semantics, which produces a sure result only when the goal is simply existential. 

% knowledge states
Interactive realizability revolves around the concept of \emph{knowledge state}. 
A knowledge state, or simply state, is a finite object that stores information about the \(\EM\) instances we use in the proof. 
The purpose of this information is to help us decide the \(\EM\) instances, that is, help us in choosing which disjunct holds. 
Moreover, whenever the state chooses the existential disjunct, it should also produce a witness, like in the BHK interpretation. 

% state representation
We can represent a state as a finite partial function\footnote{By finite partial function, we mean a partial function whose domain (the set of elements where it is defined) is finite.} that maps a concrete instance of \(\EM\) into a witness of its existential disjunct. 
Such a function decides or guesses a concrete instance \(\fa\) of \(\EM\): if it is undefined on \(\fa\), then we choose the universal disjunct; if it is defined we chose the existential disjunct with the returned witness. 
We are only interested in the instances appearing in the proof, namely, those of the form \eqref{eq:em_tot_instance} when \(\ia, \ib\) are numerals. 
Thus an instance is determined by two natural numbers; since witnesses are natural numbers too, a state can be concretely defined as a finite partial function from \(\N \times \N\) to \(\N\). 

For instance, consider the case of the \(\EM\) instances used in the proof of \gref{thm:least_element}. 
When we have to decide \eqref{eq:em_tot_instance}, we check the state on the pair \((\ia,\ib)\). 
At first, let us assume that the state is undefined on \((\ia,\ib)\). 
This means we have no knowledge about the universal disjunct \(\lrea_\ia \leR \lrea_\ib\).
Since we cannot effectively check that the universal disjunct holds, 
we make an educated guess and assume that \(\lrea_\ia \leR \lrea_\ib\) is true. 
Clearly this assumption could very well be wrong, which may or may not become apparent later in the proof. 
Keeping track of this assumption, we carry on with the proof. 
Every time we use this assumption to prove a decidable instance of its we check if the instance holds. 
More concretely, if later in the proof we use the assumption \(\lrea_\ia \leR \lrea_\ib\) to deduce that \(\lnot \op (\ib, \ia, \lpa)\) for some \(\lpa\), we check that \(\lnot \op (\ib, \ia, \lpa)\) holds. 
If this is the case, we carry on with the proof: \(\lrea_\ia \leR \lrea_\ib\) could still be false, but at least the particular instance we are using is true. 
If this is not the case, we have found a counterexample to the assumption \(\lrea_\ia \leR \lrea_\ib\): 
being negatively decidable, the counterexample is enough to effectively decide that it is false. 
Therefore we stop following the proof because we have chosen the wrong disjunct in the \(\EM\) instance \eqref{eq:em_tot_instance}. 

Moreover, a counterexample to \(\lrea_\ia \leR \lrea_\ib\) is a natural number \(\lpa\) such that \(\op (\ib, \ia, \lpa)\). 
Therefore \(\lpa\) is a witness for the existential disjunct \(\lrea_\ib \ltR \lrea_\ia\). 
We can use this new knowledge to add \((\ia,\ib)\) to the domain of the state with value \(\lpa\). Remember that we assumed the state to be undefined on \((\ia,\ib)\), which is why we assumed the universal disjunct to be true in the first place. 

At this point, we forget what we did after guessing (wrongly) that the universal disjunct was true and start again. 
More precisely, we need to backtrack to a computation state \emph{before} we decided the \(\EM\) instance in question and repeat our decision with the extended state. 
Since the extended state is defined on \((\ia,\ib)\) and yields \(\lpa\), this time we decide the \(\EM\) instance differently: we choose the existential disjunct \(\lrea_\ib \ltR \lrea_\ia\) with \(\lpa\) as witness. 
Now we are sure that our choice is the correct one and not a guess, since we have effectively decided that the existential disjunct holds (we can since it is positively decidable). 

In order for the interactive interpretation to produce correct results,
we need to assume that the state is sound, that is, when it is defined, the witness it yields is actually a witness. 
More formally, a state \(s\) is sound if, for any pair \((\ia,\ib)\), we have that \(\op(\ib, \ia, s (\ia,\ib))\) holds. 
This assumption is not problematic: the empty state, namely the state that is always undefined, satisfies it vacuously. 
Moreover, in the interactive interpretation we outlined above, we only extend a state with an actual witness. 
In other words, the extension preserves the soundness property. 

To summarize, the general procedure is the following: 
\begin{enumerate} 
  \item we start from any sound state (usually the empty state),
  \item\label{step2} we follow the proof choosing any \(\EM\) instance according to the state,
  \item if we discover that we wrongly assumed the universal disjunct of an \(\EM\) instance: 
    \begin{enumerate}
      \item we extend the state with the counterexample we found,
      \item we backtrack to a point before the \(\EM\) instance we guessed wrong,
      \item we proceed as in step \ref{step2},
    \end{enumerate}
  \item if we never discover that we wrongly assumed an universal disjunct we carry on until the end of the proof and we are done. 
\end{enumerate}

The exact point we need to backtrack to is not relevant, as long as it is before the decision of the \(\EM\) instance. 
A simple choice would be the very beginning, in which case we do not need to keep track of where we decided the \(\EM\) instance. 
In this case we only need a simple abort operator in order to formally write interactive realizers. A monadic version of this approach is given in \cite{birolo13thesis}.
A more efficient choice is to backtrack right before the decision point, so that we do not need to repeat the computations that took place before it, since they are not affected by the extension of the state. 
However this approach would require more sophisticated control operators. 

Interactive realizability can be thought as a ``smart'', albeit ``partial'', decision algorithm for negatively decidable statements. 
This can be seen comparing it with the naive algorithm given in \Cref{prog:rle}. 
It is partial because a real decision is impossible, so it only considers a finite number of instances, unlike the unbounded recursion employed by the program in \Cref{prog:rle}. 
It is smart because it does not perform a blind search, trying in order all the natural numbers. 
Instead it uses the proof itself to find the counterexamples. 
There is a reasonable expectation that the ideas underlying the proof provide a more focused way of selecting counterexamples than a blind search (this of course depends on the proof itself). 
%Moreover, if we fail to find a counterexample in the proof and if the negatively decidable statement was only used to deduce concrete instances, then we can 

Until now we considered a single instance of the \(\EM\) axiom, but little changes if there is more than one. We will return to this point later. 
In the proof of \gref{thm:least_element}, one instance of \(\EM\) is used for each inductive step in the proof. 
%  \lrea_0 \leR \lrea_1 \lor \lrea_1 \ltR \lrea_0, \lrea_1 \leR \lrea_2 \lor \lrea_2 \ltR \lrea_1, \dotsc, \lrea_{\na-1} \leR \lrea_\na \lor \lrea_\na \ltR \lrea_{\na-1}.
When we interpret the proof with the empty state, for each of these instances we assume that the universal disjunct holds. 
Therefore the proof is interpreted as follows. 
In the base step we choose \(\lrea_0\). 
In the first inductive step, we have to decide the \(\EM\) instance \( \lrea_0 \leR \lrea_1 \lor \lrea_1 \ltR \lrea_0 \). 
Since the state is empty, we assume that \(\lrea_0 \leR \lrea_1\). 
Thus we keep \(\lrea_0\) as the least element of \(\lrea_0, \lrea_1\). 
In the second inductive step, we have to decide the \(\EM\) instance \( \lrea_0 \leR \lrea_2 \lor \lrea_2 \ltR \lrea_0 \). 
Since the state is empty, we again assume that \(\lrea_0 \leR \lrea_2\). 
Thus we keep \(\lrea_0\) as the least element of \(\lrea_0, \lrea_1, \lrea_2\). 
At the end of the proof, we have assumed the following universal disjuncts:
\begin{equation}\label{eq:univ_instances}
  \lrea_0 \leR \lrea_1, \lrea_0 \leR \lrea_2, \dotsc, \lrea_0 \leR \lrea_\na. 
\end{equation}
Under these assumptions, we have found that the least element is \(\lrea_0\). 
Rather disappointing, isn't it?

The reason for this is that the universal disjuncts \(\lrea_\ia \leR \lrea_\ib\) are never instanced, so we have neither opportunity nor reason to falsify one of them. 
However this may change if \gref{thm:least_element} is used inside a bigger proof. 
This will happen later in the proof of \gref{thm:bounding_angle}. 
In this case the outer proof might instance these assumptions and discover them wrong, in which case we have to backtrack to the proof of \gref{thm:least_element}. 

Let us see how \gref{thm:least_element} behaves when its conclusion is used to deduce decidable instances. 
Assume that \(\na = 5\). 
If the state is empty, then \gref{thm:least_element} tells us that \(\lrea_0\) is a least element. 
This means that \(\lrea_0 \leR \lrea_\ia\) for any \(\ia\).
Imagine that we use \gref{thm:least_element} in a bigger proof to prove that \(\lrea_0 \leR \lrea_3\). 
This is one of the \(\EM\) instances we assumed in \eqref{eq:univ_instances}. 
Moreover, imagine that, after instantiating this assumption, we discover that \(\lrea_0 \leR \lrea_3\) does not hold at precision \(33\), that is, \(\op(\lrea_3, \lrea_0, 33)\) holds. 
Then we have to extend the domain of the state to \((0,3)\) with value \(33\). 
At this point we backtrack, say at the beginning of the proof of \gref{thm:least_element}. 

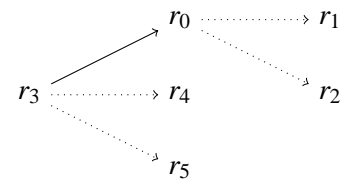
\begin{wrapfigure}{r}{0.4\textwidth}
  \vspace{-2em}
  \caption{A graph showing the result of the least element computation} \label{fig:least3}
  \begin{center}
    \begin{tikzpicture}
      \node (r3) at (0,0) {\(\lrea_3\)};
      \node (r0) at (2,1) {\(\lrea_0\)};
      \node (r4) at (2,0) {\(\lrea_4\)};
      \node (r5) at (2,-1) {\(\lrea_5\)};
      \node (r1) at (4,1) {\(\lrea_1\)};
      \node (r2) at (4,0) {\(\lrea_2\)};
      \draw [->] (r3) -- (r0);
      \draw [dotted,->] (r3) -- (r4);
      \draw [dotted,->] (r3) -- (r5);
      \draw [dotted,->] (r0) -- (r1);
      \draw [dotted,->] (r0) -- (r2);
    \end{tikzpicture}
  \end{center}
  \vspace{-1em}
  \begin{legend}
    Full arrows represent information provided by the state, dotted arrows ``guessed'' information the state knows nothing about.
  \end{legend}
  \vspace{-3mm}
\end{wrapfigure}
We again start from \(\lrea_0\) and proceed like before. 
The first and second inductive steps again select \(\lrea_0\) as the least element, assuming that \(\lrea_0 \leR \lrea_1\) and \(\lrea_0 \leR \lrea_2\). 
Things change at the third inductive step when we have to decide \(\lrea_0 \leR \lrea_3 \lor \lrea_3 \ltR \lrea_0\). 
Since now the state has a relevant witness, this time we choose the existential disjunct with witness \(33\), thus selecting \(\lrea_3\) as the new least element. 
In the next inductive steps we again assume the universal disjuncts \(\lrea_3 \leR \lrea_4\) and \(\lrea_3 \leR \lrea_5\), since the state has no information on them. 
Thus the least element is \(\lrea_3\). 
A summary of our decisions is represented in \Cref{fig:least3}.
Now imagine that we were to discover a counterexample to \(\lrea_3 \leR \lrea_2\), say at precision \(25\). 
This statement is not one of the universal disjuncts that we assumed. By looking at the proof or at \Cref{fig:least3}, we can see that it has been deduced by the semi-transitivity property from \(\lrea_3 \ltR \lrea_0\) and \(\lrea_0 \leR \lrea_2\). 
The first is the existential disjunct for which we found a witness, so we are sure that it holds. 
Thus the wrong assumption is \(\lrea_0 \leR \lrea_2\). 
By checking the proof of semi-transitivity we can see that the counterexample for \(\lrea_0 \leR \lrea_2\) is \(\max(25,33)\), thus \(33\) again. 
We extend the state accordingly and repeat the least element computation, which results in new least element \(\lrea_2\). 
In \Cref{fig:interactive_least_element} we summarize the two iterations we saw until now and add some more, as an example. 
\newcommand\stack[2][c]{{
    \renewcommand\arraystretch{0.6}
    \begin{array}{#1} 
      #2 
    \end{array}
}}
\begin{figure}[!ht]
  \caption{An example of evaluations of the interactive interpretation of \gref{thm:least_element}.} \label{fig:interactive_least_element}
  \[
    \begin{array}{|c|c|c|c|c|c|}
      \mathbf{Iter} &
      \mathbf{State} &
      \mathbf{Least\ element} & 
      \mathbf{Used} & 
      \stack{\mathbf{Deduced} \\ \mathbf{from}} & 
      \mathbf{Discovered} \\ \hline
      1st & & 
      \begin{tikzpicture}[scale=0.7,baseline=0]
        \node (r0) at (0,0) {\(\lrea_0\)};
        \node (r1) at (2,0) {\(\lrea_{1,\dotsc,5}\)};
        \draw [dotted] (r0) -- (r1);
      \end{tikzpicture}
      & 
      \lrea_0 \leR \lrea_3 & \lrea_0 \leR \lrea_3 & \lrea_3 \ltR \lrea_0 
      \\ \hline
      2nd & \lrea_3 \ltR \lrea_0
      &
      \begin{tikzpicture}[scale=0.7,baseline=-0.5]
        \node (r3) at (0,0) {\(\lrea_3\)};
        \node (r0) at (2,0.5) {\(\lrea_0\)};
        \node (r1) at (4,0.5) {\(\lrea_{1,2}\)};
        \node (r4) at (2,-0.5) {\(\lrea_{4,5}\)};
        \draw (r3) -- (r0);
        \draw [dotted] (r0) -- (r1);
        \draw [dotted] (r3) -- (r4);
      \end{tikzpicture}
      & \lrea_3 \leR \lrea_2 
      & \stack{\lrea_3 \ltR \lrea_0 \\ \lrea_0 \leR \lrea_2} 
      & \lrea_2 \ltR \lrea_0 
      \\ \hline
      3rd & \stack{\lrea_2 \ltR \lrea_0 \\ \lrea_3 \ltR \lrea_0} 
          &\begin{tikzpicture}[scale=0.7,baseline=-0.5]
      \node (r2) at (0,0) {\(\lrea_2\)};
      \node (r0) at (2,0.5) {\(\lrea_0\)};
      \node (r1) at (4,0.5) {\(\lrea_1\)};
      \node (r3) at (2,-0.5) {\(\lrea_{3,4,5}\)};
      \draw (r2) -- (r0);
      \draw [dotted] (r0) -- (r1);
      \draw [dotted] (r2) -- (r3);
    \end{tikzpicture}
    & \lrea_2 \leR \lrea_3 
    & \lrea_2 \leR \lrea_3 
    & \lrea_3 \ltR \lrea_2
    \\ \hline
    4th & \stack{\lrea_2 \ltR \lrea_0 \\ \lrea_3 \ltR \lrea_0 \\ \lrea_3 \ltR \lrea_2} 
        & \begin{tikzpicture}[scale=0.6,baseline=-0.5]
    \node (r3) at (-2,0) {\(\lrea_3\)};
    \node (r2) at (0,0.5) {\(\lrea_2\)};
    \node (r0) at (2,0.5) {\(\lrea_0\)};
    \node (r1) at (4,0.5) {\(\lrea_1\)};
    \node (r4) at (0,-0.5) {\(\lrea_{4,5}\)};
    \draw (r3) -- (r2);
    \draw (r2) -- (r0);
    \draw [dotted] (r0) -- (r1);
    \draw [dotted] (r3) -- (r4);
  \end{tikzpicture}
  & \lrea_3 \leR \lrea_1 
  & \stack{\lrea_3 \ltR \lrea_2 \\ \lrea_2 \ltR \lrea_0 \\ \lrea_0 \leR \lrea_1} 
  & \lrea_1 \ltR \lrea_0
  \\ \hline
  5th & \stack{\lrea_1 \ltR \lrea_0 \\ \lrea_2 \ltR \lrea_0 \\ \lrea_3 \ltR \lrea_0 \\ \lrea_3 \ltR \lrea_2} 
      & 
  \newcommand\safechild[1]{
      node {\(\lrea_{#1}\)}
      edge from parent[solid]
  }
  \newcommand\unsafechild[1]{
      node {\(\lrea_{#1}\)}
      edge from parent[dotted]
  }
  \begin{tikzpicture}[baseline=0,grow'=right,sibling distance=5mm,fill=white]
    \node {\(\lrea_1\)}
    child {
      node {\(\lrea_{0}\)}
      edge from parent[solid]
    }
child {
      node {\(\lrea_{2,3,4,5}\)}
      edge from parent[dotted]
    };
  \end{tikzpicture}
  & \lrea_1 \leR \lrea_4 
  & \lrea_1 \leR \lrea_4 
  & \lrea_4 \ltR \lrea_1 
  \\ \hline
  6th & \stack{\lrea_1 \ltR \lrea_0 \\ \lrea_2 \ltR \lrea_0 \\ \lrea_3 \ltR \lrea_0 \\ \lrea_3 \ltR \lrea_2 \\ \lrea_4 \ltR \lrea_1} 
      & 
  \begin{tikzpicture}[baseline=0,grow'=right,sibling distance=5mm,fill=white]
    \node {\(\lrea_4\)}
    child {
      node {\(\lrea_1\)}
      child {
        node {\(\lrea_0\)}
        edge from parent[solid]
      }
      child {
        node {\(\lrea_{2,3}\)}
        edge from parent[dotted]
      }
      edge from parent[solid]
    }
    child {
      node {\(\lrea_5\)}
      edge from parent[dotted]
    };
  \end{tikzpicture}
  & \dotso
  & \dotso
  & \dotso
  \\ \hline
\end{array}
  \]
  \begin{legend}
      {\bf Iter:} the iteration represented by the current row, 
      {\bf State:} the existential disjuncts witnessed by the state, 
      {\bf Least element:} the least element yielded by \gref{thm:least_element}, 
      {\bf Used:} a concrete consequence of \gref{thm:least_element} that is falsified in the proof, 
      {\bf Deduced from:} the premisses we deduced the falsified consequence from, 
      {\bf Discovered:} the existential assumption we found a witness of. 
  \end{legend}
\end{figure}

In general we do not use all the information in the state in each iteration: for example, in the third iteration we do not use \(\lrea_3 \ltR \lrea_0\), which we discovered in the first iteration and used in the second iteration. 
This happens because the state affects which instances of \(\EM\) are used in the proof, which should be apparent from the given iterations. 

In the iterations listed in \Cref{fig:interactive_least_element}, we compute the following sequence of least element candidates: \( \lrea_0, \lrea_3, \lrea_2, \lrea_3, \lrea_1, \lrea_4 \). 
The fact that \(\lrea_3\) appears two times may cause doubts regarding the termination of the backtracking algorithm. 
The termination of the backtracking algorithms in interactive realizability has been proven in general, see Theorem 2.15 in \cite{aschieriB10}. 

In this particular case we can understand why \(\lrea_3\) is computed two times by taking a closer look at the tree of the possible computations of the least element, 
which is shown in \Cref{fig:least_element_computation_tree}. 
For reasons of space, we only show the tree for \(\na = 3\), which is enough to see what happens up to the fifth iteration in \Cref{fig:interactive_least_element}.
\begin{figure}[!ht]
  \caption{The computation tree of the least element for \(\na = 3\)}. 
  \label{fig:least_element_computation_tree}
  \begin{center}
    \begin{tikzpicture}[level distance=10mm]
      \newcommand\leftchoice[2]{
        node {\(\lrea_{#1}\)} 
        edge from parent[dotted] 
        node[draw,left=-4mm,inner sep=1pt] 
        {\(\lrea_{#1} \leR \lrea_{#2}\)}
      }
      \newcommand\rightchoice[2]{
        node {\(\lrea_{#2}\)} 
        edge from parent[solid] 
        node[draw,right=-4mm,inner sep=1pt] 
        {\(\lrea_{#2} \ltR \lrea_{#1}\)}
      }
      \tikzstyle{level 1}=[sibling distance=6.6cm]
      \tikzstyle{level 2}=[sibling distance=3.3cm]
      \tikzstyle{level 3}=[sibling distance=2cm]
      \tikzstyle{every node}=[fill=white]
      \node {\(\lrea_0\)}
      child {
        child {
          child {
            \leftchoice{0}{3}
          }
          child {
            \rightchoice{0}{3}
          }
          \leftchoice{0}{2}
        }
        child {
          child {
            \leftchoice{2}{3}
          }
          child {
            \rightchoice{2}{3}
          }
          \rightchoice{0}{2}
        }
        \leftchoice{0}{1}
      }
      child {
        child {
          child {
            \leftchoice{1}{3}
          }
          child {
            \rightchoice{1}{3}
          }
          \leftchoice{1}{2}
        }
        child {
          child {
            \leftchoice{2}{3}
          }
          child {
            \rightchoice{2}{3}
          }
          \rightchoice{1}{2}
        }
        \rightchoice{0}{1}
      };
    \end{tikzpicture}
  \end{center}
  \begin{legend}
    Each path represents a possible computation, proceeding from root to leaf, where
    non-leaf nodes are the current least element candidates and the leaf is the final result. 
    Each branching corresponds to an \(\EM\) instance, where the left branch is taken when we guess that the universal disjunct holds for lack of information and the right branch is taken when the state contains the relevant witness. 
  \end{legend}
\end{figure}
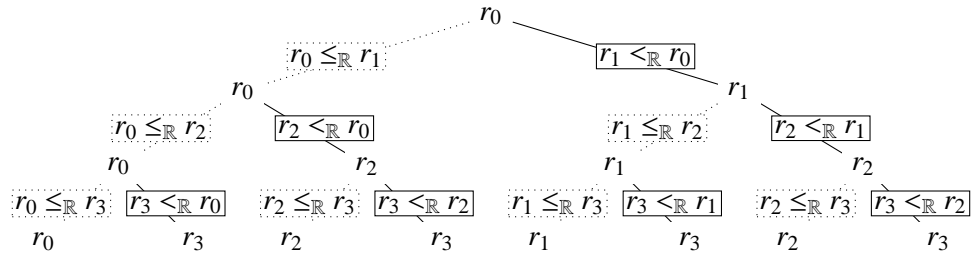
We can see that the first five iterations in \Cref{fig:interactive_least_element} correspond to the computation paths ending with the first five leaves from the left in \Cref{fig:least_element_computation_tree}, in order. 

Moreover, from the computation tree we can see that we never perform the same computation more than once. 
Indeed, assume we have just followed a particular computation path. 
When we backtrack we increment the state adding a witness of one of the \(\EM\) instances we encountered along the path, an instance we did not have a witness for. 
This means that in the next computation, when we arrive at the node corresponding to that \(\EM\) instance, instead of taking the left path as we did previously (since the state did not have a witness for that instance), we take the right path, because this time we do have a witness (since we just extended the state with it). 
Therefore, each time we backtrack, the computation path ends with a leaf that is more to the right in \Cref{fig:least_element_computation_tree}. 
This gives a bound to the number of backtrackings, namely \(2^\na - 1\).

This is very different from what one could expect by a superficial look at the proof of \gref{thm:least_element}. 
Indeed, 
if the order on the reals were decidable, then
this simple and natural proof would be quite efficient, since its complexity would be linear in \(\na\). 
However, its interactive interpretation has exponential complexity. 
This can be seen in the computation tree too: a single computation corresponds to a path and paths have length \(\na\). 
On the other hand, since we have backtracking, in the worst case we may have to perform every possible computation. 
Naturally, since the order on the reals \emph{is} undecidable, an actual comparison is impossible. 

Moreover, while in the worst case the interactive interpretation needs a time that is exponential in \(\na\), in general it is hard to estimate the amount of backtracking that will be actually performed, for two different reasons. 

    The first one is that the actual order of \(\lrea_0, \dotsc, \lrea_\na\) affects heavily the operation of the algorithm. 
Indeed, assume that \(\lrea_0\) is the least element: the interactive interpretation only performs \(\na\) dummy comparisons and immediately returns a least element candidate that, in this case, is the actual least element, so no backtracking can ensue later. 

The second reason is that the backtracking is controlled by how the least element candidate returned by the interactive interpretation is used. 
It is possible for the interactive interpretation to return a candidate that is not a least element, but such that its use in an outer proof is does not cause backtracking. 
In other words, we only need to compute a least element candidate that is good enough instead of the correct one and this can translate to a faster computation, again depending on the situation.
This also explains why the interactive interpretation is effective even if a certainly correct least element cannot be found effectively.

%\subsection{The Whole Proof Is Relevant}

In the BHK interpretation, 
the computational content of a proof is usually much longer than the algorithm it describes. 
This can be easily seen by comparing the program in \Cref{prog:min} with the proof \gref{thm:least_element}\footnote{A straightforward formalization of the proof in the Coq proof assistant is ten times longer than \Cref{prog:min}. }. 
The reason for this discrepancy is that the proof contains both the algorithm described in \Cref{prog:min} and the evidence for its correctness. 
In general, in the computational content of a proof in the BHK interpretation we can separate the part that computes values and such (the informative computation) from the part that computes the evidence showing that the values are correct (the correctness computation). 
The correctness computation does not affect the result of the informative computation and can be safely discarded when we are only interested in algorithm extraction. 

This is not the case for the computational content in the interactive interpretation. 
Here the correctness part of the computation affects the backtracking, which affects the state, which in turn affects the informative part of computation and thus the computed values. 
Therefore, in interactive realizability both parts of the proof interact to produce the final result. 
This is apparent in the second iteration, when we falsify \(\lrea_3 \leR \lrea_2\) and we have to retrace the proof of semi-transitivity, which is a non informative proof, in order to find which \(\EM\) instance we guessed wrongly and to compute the witness that we need to extend the state. 
This shows that in the interactive interpretation we cannot forget how we proved the correctness of our computations.

\section{The Real Plane}

In this section we introduce the real plane, points, lines and some relations between them. 
We use elementary analytic geometry: points are represented by coordinates, lines by equations and proofs are mostly computations with real numbers. 

We represent a point as a pair of real numbers, its coordinates. 
Formally we can say that a point is just a natural number \(\ia\) and that there is a primitive recursive function mapping indexes into pairs of real numbers. 
As we did for real numbers, in order to improve readability we add some sugar to the notation and 
use the metavariables \(\pa, \pb, \pc, \pd\) for arithmetic terms used as indexes of points. 
When we use the index of a point both as a number and as a point, we write it as \(\ia\) in the first case and as \(\pa_\ia\) in the second. 
We write the coordinates of a point \(\pa\) as \((x_\pa, y_\pa)\) and of a point \(\pa_\ia\) as \((x_\ia, y_\ia)\). 
A line passing through two points \(\pa\pb\) is written as \(\pa\pb\).

%\subsection{Operations on Real Numbers}

\newcommand\Rzero{0_\R}
\newcommand\Rplus{+_\R}
\newcommand\Rminus{-_\R}
\newcommand\Rprod{\cdot_\R}

Before proceeding we need to introduce further infrastructure for the real numbers. 
Any rational number \(\lqa\) can be represented as a real number by taking the constant sequence of the degenerate interval \([\lqa, \lqa]\). 
Let \(\Rzero\) be the representation of the rational zero. 
We can define addition, subtraction and multiplication on the nested interval sequences by using the corresponding rational operation point-wise on the extremes. 
It is possible to retain the exponential convergence by taking a suitable sub-sequence. This can be done effectively and follows from the continuity of the operations on the rationals. 
Thus we can safely assume that we have addition, subtraction and multiplication on the reals. 

%\subsection{The Left and Right Predicates} 

\newcommand\leftP{\operatorname{\sf left}}
\newcommand\rightP{\operatorname{\sf right}}
In order to write the formal statement of \gref{thm:bounding_angle}, we need a way to determine the position of a point with respect to a line. 

First of all consider two points \(\pa\) and \(\pb\). 
We can write the equation that a point \(\pc\) has to satisfy to be on the line going through them:
\begin{equation} \label{eq:line}
  (x_\pb - x_\pa)(y_\pc - y_\pa) - (x_\pc - x_\pa)(y_\pb - y_\pa) \eqR \Rzero. 
\end{equation}
If the left-hand side is zero then \(\pc\) is on the same line with \(\pa\) and \(\pb\). 
When left-hand side is not zero, we can use its sign to distinguish which side of \(\pa\pb\) \(\pc\) is on. 
We call these sides left and right. 
We write \( \leftP (\pa, \pb, \pc) \) (resp. \( \rightP (\pa, \pb, \pc) \)) and we say that \(\pc\) is to the \emph{left} (resp. \emph{right}) of the line passing through the points \(\pa\) and \(\pb\) when 
\begin{gather*}
  \leftP(\pa, \pb, \pc) \equiv (x_\pb - x_\pa)(y_\pc - y_\pa) - (x_\pc - x_\pa)(y_\pb - y_\pa) \gtR \Rzero, \\
  \rightP(\pa, \pb, \pc) \equiv (x_\pb - x_\pa)(y_\pc - y_\pa) - (x_\pc - x_\pa)(y_\pb - y_\pa) \ltR \Rzero. 
\end{gather*}
Both \(\leftP\) and \(\rightP\) are positively decidable, since they are defined by means of \(\ltR\). 
Moreover, note that \(\pc\) is to the left of \(\pa \pb\) if and only if \(\pb\) is to the right of \(\pa\pc\). 
We say that \(\pa\) is \emph{above} \(\pb\) if \( y_\pa \geR y_\pb \) and that \(\pc\) is \emph{below} \(\pb\) when \( y_\pc \leR y_\pb \).

\section{The Geometric Part of the Proof}

Now we are ready to present the rest of the proof of the main statement. 
We divide the proof in two parts, the first given as a lemma. 
Since these proofs are more complex, for reason of readability and space 
we will not be as formal as we have been until now. 

From this point onward we assume that no three points are on the same line, formally:
\begin{equation} \label{eq:no_three_aligned}
  \qforall{\pa,\pb,\pc} \leftP (\pa, \pb, \pc) \lor \rightP (\pa, \pb, \pc). 
\end{equation}
This is a strong assumption, even more so because we require this to hold constructively: since \(\leftP\) and \(\rightP\) are \(\Sigma^0_1\) formulas defined with \(\leR\), we assume that we have an effective map that given three points yields the precision we need to reach in order to check that \(\pc\) is not on the line \(\pa\pb\). 
In other words, we are assuming that we have a procedure that effectively decides instances of the \(\leftP\) and \(\rightP\) predicates. 
The effective computation we extract uses this procedure as a parameter. 

A further consequence is that all points must be distinct: when \(x_\pa \eqR x_\pb\) and \(y_\pa \eqR y_\pb\), the left-hand side in \eqref{eq:line} is always zero for any \(\pc\).

In the next lemma the points \(\pb_0, \pb_1, \pb_2\) are three generic points, 
that is, \(\pb_\ia\) is not necessarily the point indexed by the natural number \(\ia\). 
Moreover we assume that the index \(\ia\) in \(\pb_\ia\) is interpreted up to congruence modulo 3 and thus always falls in \(\{0,1,2\}\). 
For instance, when we write \(\pb_4\), we actually mean \(\pb_1\). 
We write the coordinates of \(\pb_\ia\) as \((x_\ia,y_\ia)\), with the same conventions for the index. 
We prove that when three points are one to the left of the other with respect to a central one, one of them is necessarily lower than the central point. 
\begin{lemma}[Three points] \label{thm:three_points}
  Assume \eqref{eq:no_three_aligned} and let \(\pa, \pb_0, \pb_1\) and \(\pb_2\) be four points in the real plane such that \(\pb_{\ia+1}\) is to the left (resp. right) of \(\pa \pb_\ia\) for any \( \ia < 3\).
  Then at least one of \(\pb_0, \pb_1, \pb_2\) is strictly below \(\pa\). 
  Formally:
  \[
    \qforall{\pa, \pb_0, \pb_1, \pb_2} 
    (\qforall{\ia < 3} \leftP(\pa,\pb_\ia,\pb_{\ia+1}))
    \limply
    \qexists{\ia < 3} y_\ia \ltR y_\pa.
  \]
\end{lemma}
% \begin{figure}[!ht]
%   \caption{The three points lemma when \(\pb_2\) is the point below \(\pa\).}
%   \label{fig:three_points}
%   \begin{center}
%     \begin{tikzpicture}[>=latex] % angle from
%       \draw[fill] (0,0) circle (0.05);
%       \node[right] at (0,0) {\(\pa\)};
%       \draw[fill] (45:2) circle (0.05);
%       \node[left] at (45:2) {\(\pb_0\)};
%       \draw (-135:0.5) -- (45:2.5);
%       \draw[fill] (140:1.5) circle (0.05);
%       \node[above] at (140:1.5) {\(\pb_1\)};
%       \draw (-40:0.5) -- (140:2);
%       \draw[fill] (240:1) circle (0.05);
%       \node[right] at (240:1) {\(\pb_2\)};
%       \draw (60:0.5) -- (240:1.5);
%       \draw[dashed] (-2,0) -- (2,0);
%     \end{tikzpicture}
%   \end{center}
% \end{figure}
We omit the proof for reasons of space. Since it is a constructive proof, its interactive interpretation coincides with its BHK interpretation and thus is not particularly relevant for our analysis. 

We can now prove the main statement. 
\begin{theorem}[Convex Angle] \label{thm:bounding_angle}
  Assume \eqref{eq:no_three_aligned}. For any \(\na \ge 2\), we can select three points \(\pa\), \(\pb\) and \(\pc\) from \(\{\pa_0, \dotsc, \pa_\na\}\) 
  such that all the remaining points fall in the angle \(\widehat{\pb\pa\pc}\), that is, all points are to the left of \(\pa\pb\) and to the right of \(\pa\pc\). 
  \begin{multline*}
    \qforall{\na\ge 2} \qexists{\ia,\ib,\ic \le \na} 
    \qforall{\id \le \na} \id \neq \ia \limply 
    (\id \neq \ib \limply \leftP (\pa_\ia, \pa_\ib, \pa_\id)) \land
    (\id \neq \ic \limply \rightP (\pa_\ia, \pa_\ic, \pa_\id)).
  \end{multline*}
\end{theorem}
Note that the convexity of the angle \(\widehat{\pb\pa\pc}\) is assured, because we require that \(\pc\) is to the left of \(\pa\pb\) and \(\pb\) to the right of \(\pa\pc\).
\begin{proof}[Classical proof]
  Let \(\pa\) be the point with the least vertical coordinate and choose other two points \(\pb'\) and \(\pc'\), which are our candidates for \(\pb\) and \(\pc\) respectively. 
  We want all points except \(\pa\) to be to the left of \(\pa\pb\) and to the right of \(\pa\pc\). 
  If \(\pb'\) is to the left of \(\pa\pc'\), we swap \(\pb'\) and \(\pc'\). 
  Thus we know that \(\pb'\) is to the right of \(\pa\pc'\) and \(\pc'\) is to the left of \(\pa\pb'\). 

  Now consider any point \(\pd\) except \(\pa\), \(\pb'\) and \(\pc'\). 
  We have four cases:
  \begin{itemize}
    \item 
      If \(\pd\) is to the left of \(\pa \pb'\) and if it is to the right of \(\pa \pc'\), then we keep \(\pb'\) and \(\pc'\) as candidates for \(\pb\) and \(\pc\). 
    \item 
      If \(\pd\) is to the right of \(\pa \pb'\), then we choose \(\pd\) as the new candidate for \(\pb\).

      \begin{minipage}{0.65\linewidth}\indent
        Clearly \(\pb'\) is to the left of \(\pa \pd\). 
        Moreover, any other point \(\pd'\), which we already checked to be to the left of \(\pa \pb'\), is to the left of \(\pa \pd\) too. 
        This is a consequence of \eqref{eq:no_three_aligned} and \Cref{thm:three_points}. 

        Indeed, from \eqref{eq:no_three_aligned}, we know that \(\pd'\) is either to the left or to the right of \(\pa\pd\). 
        We already know that \(\pd\) is to the right of \(\pa \pb'\) and \(\pb'\) is to the right of \(\pa \pd'\).
      \end{minipage}\begin{minipage}[c]{0.35\linewidth}
        \begin{center}
          \begin{tikzpicture}[>=latex] 
            \pointAt{A}{0,0}\pa{below}
            \pointAt{B}{60:1}{\pb'}{right}
            \lineBetween{A}{B}
            \pointAt{C}{150:1.3}{\pc'}{left}
            \lineBetween{A}{C}
            \pointAt{D}{1,0.5}{\pd}{right}
            \lineBetween{A}{D}
            \pointAt{D'}{-1,1}{\pd'}{right}
            \draw[dashed] (-2,0) -- (2,0);
          \end{tikzpicture}
        \end{center}
      \end{minipage}\\
        If \(\pd'\) were to the right of \(\pa \pd\), then by \Cref{thm:three_points}, one of \(\pb'\), \(\pd\) or \(\pd'\) would have been strictly lower than \(\pa\) which would be a contradiction, since \(\pa\) is the lowest point.
        Thus \(\pd'\) is to the left of \(\pa \pd\). 
    \item 
      Symmetrically, if \(\pd\) is to the left of \(\pa \pc'\), then we choose \(\pd\) as the new candidate for \(\pc\). 
    \item 
      We show that \(\pd\) cannot be to the right of \(\pa \pb'\) and to the left of \(\pa \pc'\). 

\begin{minipage}[c]{0.65\linewidth}
      If this were the case, \(\pb'\) would be to the left of \(\pa \pd\) and \(\pd\) would be to the left of \(\pa \pc'\). 
      Since we know that \(\pc'\) is to the left of \(\pa\pb'\), by \Cref{thm:three_points}, one of \(\pd\), \(\pb'\) or \(\pc'\) would be strictly lower than \(\pa\). 
      This is a contradiction, since \(\pa\) is the lowest point by \gref{thm:least_element}. 
      \end{minipage}\begin{minipage}[c]{0.35\linewidth}
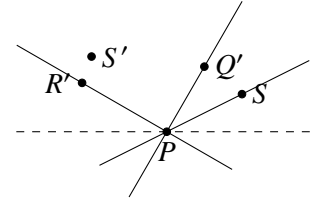
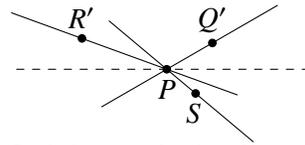
      \begin{center}
        \begin{tikzpicture}[>=latex] 
          \pointAt{A}{0,0}\pa{below}
          \pointAt{B}{30:0.7}{\pb'}{above}
          \lineBetween{A}{B}
          \pointAt{C}{160:1.2}{\pc'}{above}
          \lineBetween{A}{C}
          \pointAt{D}{-40:0.5}{\pd}{below}
          \lineBetween{A}{D}
          \draw[dashed] (-2,0) -- (2,0);
        \end{tikzpicture}
      \end{center}
\end{minipage}
  \end{itemize}
  We repeat this procedure for all the points except \(\pa\), \(\pb'\) and \(\pc'\) and we find the required points \(\pb\) and \(\pc\). 
\end{proof}

For convenience we have written the proof as an iterative algorithm. 
The proof is actually by induction on a slightly stronger version of the final statement, that adds the requirement for \(\pa\) to be lower than all the other points. 

\section{The Interactive Interpretation}

Before studying the interactive interpretation of the whole proof of \gref{thm:bounding_angle} along with its lemmas, we need to understand their computational significance. 
Thus we stop for a moment and recall some general considerations on the computational meaning of formulas in the BHK interpretation and, more specifically, in the Curry-Howard correspondence. 

As a consequence of the proofs-as-programs and formulas-as-types interpretation, 
the conclusion of a proof (that is, the statement it proves) can be thought of as the specification of the program representing the proof. 

%\subsection{Subroutines, arguments and effective computations}

In order to understand how the interactive interpretation works, it is important to distinguish computations that can be carried out effectively from those that cannot. 
Consider a proof of a statement of the form \(\qforall\lva \qexists\lvb \fa\).  
If we read it as a specification, it calls for a program that describes a function, a subroutine.
It takes a natural number as an argument named \(\lva\) and returns a pair containing a natural number \(\lvb\) and a program/proof that \(\lvb\) satisfies \(\fa\). 
All of our theorems begin with universal quantifications and implications, that is, they are specification for programs describing functions with arguments. 
Thus, in order to have an actual computation we have to provide the program with the required arguments. 

%\subsection{The Interactive Interpretation of the Whole Proof}

We can now explain the interactive interpretation of the whole proof, composed of the two lemmas and the final algorithmic proof. 
We focus on the interaction between these parts without analyzing each part in detail as we have done for \gref{thm:least_element}. 
%We begin by analyzing the single lemmas and their significance in the proof. 

We start by considering the statement of \gref{thm:bounding_angle}. 
Assume that we are given a natural number \(\na\). 
In the proof we work with the first \(\na+1\) points of the enumeration. 

The proof is an iterative procedure to select \(\pa\), \(\pb\) and \(\pc\) satisfying the following \emph{bounding condition}: 
\begin{equation} \label{eq:bounding_condition}
  \qforall{\id \le \na} \id \neq \ia \limply
  (\id \neq \ib \limply \leftP (\pa_\ia, \pa_\ib, \pa_\id)) \land
  (\id \neq \ic \limply \rightP (\pa_\ia, \pa_\ic, \pa_\id)). 
\end{equation}
The bounding condition specifies an informative computation, since \(\leftP\) and \(\rightP\) are defined by means of \(\ltR\), which is an existential quantification. 
Thus its proof computes some witnesses, namely the precision of the comparisons we need to check that the bounding condition holds. 
While we are mainly interested in the choice of the points \(\pa, \pb\) and \(\pc\) and not in the information needed to prove the bounding condition itself, 
the precision of the computation provided by \eqref{eq:bounding_condition} is actually used in interactive interpretation since it can cause backtracking. 

We claim that this bounding condition specifies an effective computation. 
First of all, the outer universal quantification is bounded, thus, in order to compute the condition, we have to compute the body of the quantification \(\na+1\) times. 
The same holds for the conjunctions. 
Thus the effectiveness of the whole condition follows from the effectiveness of the conjuncts. 
%The inequalities are atomic, so their proofs are irrelevant\footnote{This would not be the case if they were deduced by \(\EM\).}. 
The implications are effective: 
their only argument, the proof of the antecedent, is arithmetical atomic, hence irrelevant, thus the computations they specify must be constant functions. 
Therefore, we can effectively compute them by applying them to any single argument. 
Finally their consequents specify effective computations, thanks to \eqref{eq:no_three_aligned}, the assumption that no three points are on the same line. 
Thus, proofs of the bounding condition describe effective computations.

%This is very important in the interactive interpretations, since it guarantees that the correctness computation returned by \gref{thm:least_element} is going to be instanced, that is, it is going to be applied to the values needed to check the bounding condition. 
%If the correctness computation fails for one or more value, then we backtrack and compute a new lowest point \(\pa\) with the new information we discovered. 

%We can see in more detail what happens by studying the proof of \gref{thm:bounding_angle}. 
Now we can start following the proof. 
In the beginning, the lowest point \(\pa\) is selected using \gref{thm:least_element} on the vertical coordinate. 
Consider the statement of \gref{thm:least_element}:
\[
  \qforall\na \qexists{\ia \le \na} \qforall{\ib \le \na} \lrea_\ia \leR \lrea_\ib. 
\]
As a specification, it calls for a program that, given \(\na\), yields the value \(\ia\) and the correctness computation that checks that \(\ia\) is the least element. 
Since the correctness computation cannot be carried out effectively (it is negatively decidable),
the interactive interpretation computes a trivial least element the first time.
If later in the proof we happen to partially compute the correctness computation, then we may discover new information and backtrack again to the least element computation. 
%Once \(\na\) is given, the bounded quantification behaves like a conjunction and the condition is negatively decidable. 
Since \gref{thm:least_element} does not necessarily return a least element, but only a least element candidate, \(\pa\) is not the lowest point either, but just a lowest point candidate. 

The role of \Cref{thm:three_points} is to prove that some point is strictly lower than \(\pa\), thus producing a contradiction. 
In the classical proof this ensures that undesirable situations never happen. 
In the interactive interpretation however, since \(\pa\) is not necessarily the lowest point, no contradiction occurs. 
Instead, what happens is that we actually are in one of the cases we had excluded in the classical proof. 
At this point, in order to deduce the contradictory statement, we have partially computed the correctness computation returned by \gref{thm:least_element} and thus discovered which assumption was incorrectly guessed. 
We compute the relevant witness and extend the state accordingly. 
Then we compute a new lowest point candidate and continue again following the proof of \gref{thm:bounding_angle} until either we can satisfy its conclusion or we backtrack again. 

We use \Cref{thm:three_points} in two places in the proof of \gref{thm:bounding_angle}. 
The first use takes place when, while iterating on the points, we discover that the bounding condition fails for some \(\pd\) and we choose \(\pd\) as the new candidate for \(\pb\) or \(\pc\). 
We use \Cref{thm:three_points} to show that this choice satisfies the bounding condition for all the previous points we iterated over until now. 
More precisely we use \Cref{thm:three_points} to prove that, if the bounding conditions fails for \(\pd\), then one of \(\pb\), \(\pc\) or \(\pd\) is strictly lower than \(\pa\).
As we described previously, this in turn starts the backtracking. 

We also use \Cref{thm:three_points} to claim that the bounding condition cannot fail because \(\pd\) is both to the right of \(\pa\pb\) and to the left of \(\pa\pc\). 
This case was excluded completely in the classical proof, since it always leads to contradiction. 
When it occurs in the interactive interpretation, we backtrack for sure since the bounding condition cannot be satisfied. 
More precisely, in this case \Cref{thm:three_points} proves that one of \(\pb\), \(\pc\) or \(\pd\) is strictly lower than \(\pa\). 
Therefore, in order to get the contradiction, we instantiate the assumptions \(y_\pa \leR y_\pb\), \(y_\pa \leR y_\pc\) and \(y_\pa \leR y_\pd\) with enough precision to falsify at least one of them. 

As a last example, consider a situation where the state is empty and thus \(\pa\) is simply the first point in the enumeration. Assume that the points are arranged as shown in \Cref{fig:no_backtracking}.
\begin{wrapfigure}{r}{5cm}
  \vspace{-6mm}
  \caption{A situation where no backtracking occurs.} \label{fig:no_backtracking}
  \vspace{3mm}
  \begin{center}
    \begin{tikzpicture}[>=latex]
      \pointAt{A}{0,0}\pa{below}
      \pointAt{B}{-140:1.5}\pb{above}
      \pointAt{C}{-30:1}\pc{above}
      \pointAt{D}{-60:1.5}\pd{below}
      \pointAt{D1}{-100:1}{}{}
      \pointAt{D2}{-130:1.2}{}{}
      \lineBetween{A}{B}
      \lineBetween{A}{C}
      \draw[dashed] (-2,0) -- (2,0);
    \end{tikzpicture}
  \end{center}
  \vspace{-6mm}
\end{wrapfigure}
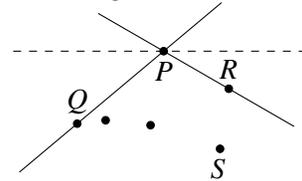
Since the bounding condition is satisfied immediately, we never need to use \Cref{thm:three_points}. Thus backtracking never ensues. 
This mean that \(\pa\), while certainly not the lowest point, is a good enough candidate and we do not need another one. 
This is one of the cases we mentioned where the interactive interpretation produces a fast computation, since the lowest point is only computed once and the proof ends with no backtracking. 
This shows how the behavior of interactive interpretation of \gref{thm:least_element} depends heavily on the final statement of the proof.

%\section{Bibliography}

\bibliographystyle{eptcs}
\bibliography{realizability}

\end{document}